\documentclass[conference]{IEEEtran}
\IEEEoverridecommandlockouts
\usepackage{cite}
\usepackage{amsmath,amssymb,amsfonts}
\usepackage{algorithmic}
\usepackage{amsthm}              
\usepackage{graphicx}
\usepackage{textcomp}
\usepackage{xcolor}
\usepackage{bm}
\usepackage{eucal}
\usepackage{mathrsfs}
\def\BibTeX{{\rm B\kern-.05em{\sc i\kern-.025em b}\kern-.08em
    T\kern-.1667em\lower.7ex\hbox{E}\kern-.125emX}}

\newtheorem{theorem}{Theorem}

\usepackage{etoolbox}     
\makeatletter
\def\myaffilspace{-0.35em}  
\patchcmd{\@IEEEauthorblockN}{\\}{\\[\myaffilspace]}{}{}
\patchcmd{\@IEEEauthorblockA}{\\}{\\[\myaffilspace]}{}{}
\makeatother

\begin{document}

\title{Low-Latency Terrestrial Interference Detection for Satellite-to-Device Communications}

\author{%
  \IEEEauthorblockN{%
    Runnan~Liu\IEEEauthorrefmark{1},
    Weifeng~Zhu\IEEEauthorrefmark{2},
    Shu~Sun\IEEEauthorrefmark{3},
    Wenjun~Zhang\IEEEauthorrefmark{3}}
  \\[-1em]
  \IEEEauthorblockA{\IEEEauthorrefmark{1} China Academy of Information and Communications Technology, Beijing 100191, China}
  \\[-1em]
  \IEEEauthorblockA{\IEEEauthorrefmark{2} Department of Electrical and Electronic Engineering, The Hong Kong Polytechnic University, Hong Kong SAR, China}
  \\[-1em]
  \IEEEauthorblockA{\IEEEauthorrefmark{3}School of Information Science and Electronic Engineering, Shanghai Jiao Tong University, Shanghai 200240, China}
}

\maketitle

\begin{abstract}
Direct satellite-to-device communication is a promising future direction due to its lower latency and enhanced efficiency. However, intermittent and unpredictable terrestrial interference significantly affects system reliability and performance. Continuously employing sophisticated interference mitigation techniques is practically inefficient. Motivated by the periodic idle intervals characteristic of burst-mode satellite transmissions, this paper investigates online interference detection frameworks specifically tailored for satellite-to-device scenarios. We first rigorously formulate interference detection as a binary hypothesis testing problem, leveraging differences between Rayleigh (no interference) and Rice (interference present) distributions. Then, we propose a cumulative sum (CUSUM)-based online detector for scenarios with known interference directions, explicitly characterizing the trade-off between detection latency and false alarm rate, and establish its asymptotic optimality. For practical scenarios involving unknown interference direction, we further propose a generalized likelihood ratio (GLR)-based detection method, jointly estimating interference direction via the Root-MUSIC algorithm. Numerical results validate our theoretical findings and demonstrate that our proposed methods achieve high detection accuracy with remarkably low latency, highlighting their practical applicability in future satellite-to-device communication systems.
\end{abstract}

\begin{IEEEkeywords}
satellite communications, interference detection, change detection, generalized likelihood ratio, CUSUM algorithm, Root-MUSIC
\end{IEEEkeywords}

\section{Introduction}
Satellite communications have become an essential component of the evolving wireless network ecosystem, playing a central role in emerging paradigms such as sixth-generation (6G) wireless systems and non-terrestrial networks (NTNs) \cite{chen2020systemintegration, zhu2022integrated, lin2021pathto6g}.
Depending on how satellite signals are delivered to user devices, satellite communication systems can be broadly categorized into two distinct architectures: conventional relay-based satellite communications and emerging direct satellite-to-device communications \cite{bakhsh2025multiSatelliteMIMO}.
Compared with the relay-based approach, \textit{direct satellite-to-device} technology considerably reduces transmission latency and overall system complexity, while significantly enhancing mobility support and communication efficiency \cite{pasandi2024survey}.
Nevertheless, due to the direct exposure of user devices to terrestrial environments and the inherently limited antenna gain of handheld terminals, direct satellite-to-device services are inherently susceptible to various sources of \textit{terrestrial interference} \cite{sharma2012satellitecognitive}, such as radar emissions \cite{zheng2019radarcommunicationcoexistence}, ambient electromagnetic radiation \cite{ghile2016influence}, and cosmic radiation \cite{hoeffgen2020investigating}.
In practice, such interference signals are prevalent, diverse, and inherently random, substantially degrading the signal-to-interference-plus-noise ratio (SINR) and thereby undermining the performance and reliability of satellite communication links \cite{sharma2012satellitecognitive}.

To mitigate terrestrial interference, numerous sophisticated approaches, such as adaptive beamforming and interference cancellation techniques, have been extensively studied in the literature \cite{cottatellucci2006interference,peng2022integrating,xie2023interference}.
However, advanced interference mitigation algorithms typically involve high computational complexity and substantial operational costs.
Furthermore, terrestrial interference signals usually occur \textit{sporadically and unpredictably}, rendering the permanent use of complex interference mitigation techniques inefficient and impractical.
Consequently, an operationally preferable strategy is to first develop an efficient interference detector that promptly and accurately identifies interference occurrences, subsequently triggering sophisticated interference mitigation and localization methods only upon positive detection.
Although crucial from practical perspectives, interference detection in direct satellite-to-device communications has not received sufficient attention in previous literature.

In this paper, we conduct a rigorous and comprehensive investigation into online terrestrial interference detection, specifically tailored for direct satellite-to-device communication systems. Capitalizing on the unique operational characteristics of satellite communication, where satellites typically operate in a periodic burst-mode fashion consisting of alternating burst transmission and idle phases, we recognize that signals received at the user equipment (UE) during idle phases exclusively indicate terrestrial interference. Motivated by this insight, we propose detection algorithms that explicitly utilize received signals in these idle intervals, thereby enabling rapid and accurate identification of interference events.

First, we formulate the interference detection task as a binary hypothesis testing problem, explicitly characterizing the statistical differences in received signal amplitudes (Rayleigh distribution under interference-free conditions versus Rice distribution in the presence of interference). Under the simplified scenario where interference direction is known, we develop an online detection scheme based on the cumulative sum (CUSUM) algorithm, and analytically characterize the fundamental trade-off between detection delay and false alarm rate. Notably, we establish analytical guarantees on the asymptotic optimality of the proposed CUSUM-based approach. In the more practical case where the interference direction is unknown, we further propose a generalized likelihood ratio (GLR)-based detection method, whereby the unknown interference direction is jointly estimated using a computationally efficient and high-resolution Root- MUltiple SIgnal Classification (Root-MUSIC) algorithm. Extensive numerical results validate the robustness and effectiveness of our proposed methods, demonstrating that the theoretical predictions closely match our simulated performance. Remarkably, these algorithms successfully achieve impressively low average detection delays even in challenging low-interference scenarios, highlighting their practical utility for realistic deployments in satellite-direct-to-phone communication systems.

\section{System Model}

Consider a satellite-to-device communication scenario where the UE, equipped with an $M$-element antenna array, directly receives signals from a satellite. As illustrated in Fig.~\ref{fig:ServiceCycle}, the satellite operates in periodic service cycles, each consisting of a burst transmission phase followed by an idle phase. In the $k$-th service cycle, the burst transmission phase (shown as the high-level segment) delivers downlink packets to the UE, while the subsequent idle phase (shown as the low-level segment) corresponds to periods when the satellite ceases transmissions. During these idle intervals, any received signal at the UE must originate from external interference sources, since the satellite itself is silent. 
\textit{This unique feature makes the idle phase particularly advantageous for interference detection.}
By leveraging this property, our approach specifically monitors the UE’s received signal during idle phases to enable rapid and accurate detection of external interference.
\begin{figure}[t]
  \centering
    \includegraphics[width=0.9\linewidth]{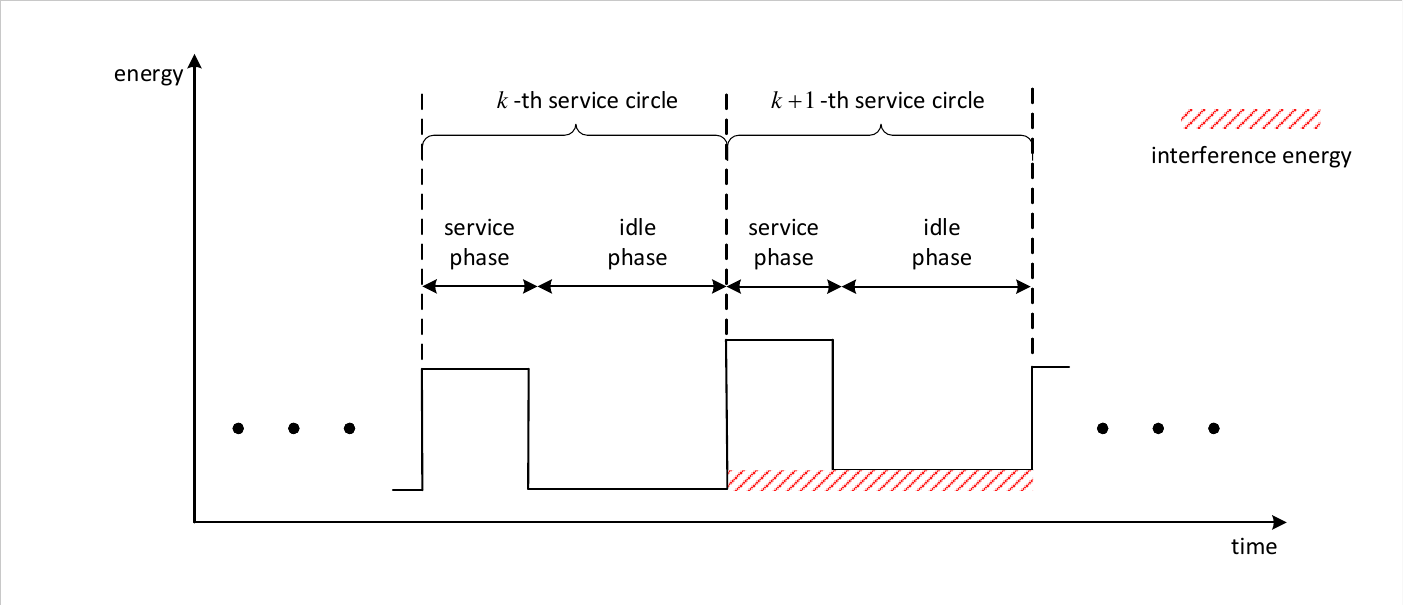}
  \caption{Illustration of one service cycle: burst transmission phase followed by idle phase.}
  \label{fig:ServiceCycle}
\end{figure}

Let $\mathbf{y}_k\in\mathbb{C}^{M\times1}$ denote the received signal vector at the UE during the idle phase of the $k$-th service interval.
In practice, two distinct scenarios may arise.
In the absence of interference, the received signal consists purely of additive white Gaussian noise (AWGN), expressed as $\mathbf{y}_k=\mathbf{n}_k,~\forall k>0$, where $\mathbf{n}_k\sim \mathcal{CN}(0,\sigma_n^2\mathbf{I}_M)$ denotes the AWGN vector with covariance $\sigma_n^2\mathbf{I}_M$.
Alternatively, when interference appears at an unknown change-point, denoted as $\nu$, the received signal is given by $\mathbf{y}_k=\mathbf{n}_k,~k<\nu$ and $\mathbf{y}_k=\mathbf{t}_k+\mathbf{n}_k,~k\ge\nu$, where $\mathbf{t}_k\in \mathbb{C}^{M \times 1}$ represents the interference signal vector.
This results in a binary hypothesis testing problem, formulated as
\begin{equation}
\begin{aligned}
&H_0~(\text{no interference}): &&\mathbf{y}_k=\mathbf{n}_k, &&\forall k>0, \\
&H_1~(\text{with interference}): &&\mathbf{y}_k=\mathbf{n}_k, &&0 < k < \nu,\\
& &&\mathbf{y}_k=\mathbf{t}_k+\mathbf{n}_k, &&k\ge\nu.
\end{aligned}
\label{Eq_Goal}
\end{equation}

In this paper, we model the interference signal $\mathbf{t}_k$ as a plane-wave impinging on the antenna array from a specific direction of arrival (DoA), denoted by $\theta^{(I)}$. This assumption is practical in satellite communications, given the large distances involved, resulting in approximately planar wavefronts. Thus, the interference vector can be modeled as a scaled steering vector
\begin{equation}
\mathbf{t}_k = \sigma^{(I)} e^{j\phi_k} \mathbf{a}(\theta^{(I)}),
\end{equation}
where $\sigma^{(I)}$ denotes the interference amplitude, and the random phase $\phi_k\sim\textit{U}(0,2\pi)$ denotes typical random phase variations. Considering a uniform linear array (ULA) at the UE, the steering vector $\mathbf{a}(\theta)$ can be expressed as
\begin{equation}
\mathbf{a}(\theta) = \frac{1}{\sqrt{M}}\begin{bmatrix}
1 & e^{j\pi \sin(\theta)} & \dots & e^{j\pi(M-1) \sin(\theta)}
\end{bmatrix}^{\mathrm{T}},
\end{equation}
where the normalization factor $1/\sqrt{M}$ ensures unit norm.

To simplify the detection task and enhance sensitivity, we project the received signal vector $\mathbf{y}_k$ onto the steering vector $\mathbf{a}(\theta^{(I)})$, yielding a scalar test statistic
\begin{equation}
z_k = \mathbf{a}(\theta^{(I)})^H \mathbf{y}_k.
\end{equation}
The amplitude $r_k = |z_k|$ serves as the decision metric. This projection is effectively a \textit{matched filter} operation, widely recognized in radar, sonar, and communication systems as an optimal linear detector under AWGN conditions, maximizing the output SNR.

Under hypothesis $H_0$, the statistic $z_k$ is complex Gaussian distributed as $\mathcal{CN}(0, \sigma_n^2)$, resulting in a Rayleigh-distributed amplitude $r_k$ with probability density function (PDF)
\begin{equation}
f_{R}^{(0)}(r) = \frac{2r}{\sigma_n^2} \exp\left(-\frac{r^2}{\sigma_n^2}\right), \quad r \geq 0.
\label{PDF_R_0}
\end{equation}
Under hypothesis $H_1$, the presence of interference introduces a deterministic component, leading $r_k$ to follow a Rice distribution
\begin{equation}
f_{R}^{(1)}(r)=\frac{2r}{\sigma_n^2}\exp\left(-\frac{r^2+(\sigma^{(I)})^2}{\sigma_n^2}\right)I_0\left(\frac{2\sigma^{(I)}r}{\sigma_n^2}\right), \quad r \geq 0,
\label{PDF_R_1}
\end{equation}
where $I_0(\cdot)$ denotes the modified Bessel function of the first kind and order zero.

Consequently, the interference detection task defined in (\ref{Eq_Goal}) simplifies to distinguishing between Rayleigh-distributed amplitudes under $H_0$ and Rice-distributed amplitudes under $H_1$, formally as
\begin{equation}
r_k \sim
\begin{cases}
f_{R}^{(0)}(r), & k < \nu, \\
f_{R}^{(1)}(r), & k \geq \nu.
\end{cases}
\end{equation}
Our objective is to develop efficient detection algorithms that minimize the detection delay while maintaining a predefined false alarm rate (FAR).

\section{On-line Interference Change Detection with Known Direction}
In this section, we consider the scenario where the interference direction $\theta^{(I)}$ is known a priori. This idealized setting serves as a performance benchmark and offers theoretical insights for the more practical case in which the interference direction is unknown.

\subsection{The CUSUM-based Detector for Known Interference Direction}
At each service cycle $k$, we define the log-likelihood ratio (LLR) between the event $r_k \sim f_{R}^{(1)}(r)$ and $r_k \sim f_{R}^{(0)}(r)$ as
\begin{equation} 
\ell(r_k) = \log \frac{f_R^{(1)}(r_k)}{f_R^{(0)}(r_k)},~\forall k. 
\label{LLR}
\end{equation} 
Next, substituting (\ref{PDF_R_0}) and (\ref{PDF_R_1}) into (\ref{LLR}), we can obtain
\begin{equation} 
\ell(r_k) = \log I_0\left(\frac{2\sigma^{(I)} r_k}{\sigma_n^2}\right) - \frac{(\sigma^{(I)})^2}{\sigma_n^2}, \quad \forall k,
\end{equation} 
where $I_0(\cdot)$ denotes the modified Bessel function of the first kind with order zero.

A straightforward detection approach is to make a decision based solely on each $\ell(r_k),~k>0$. 
However, this sample-by-sample method is highly sensitive to noise, leading to poor detection accuracy, particularly under low SNR conditions or subtle changes.
To face this challenge, we introduce a CUSUM-based interference detector.
Specifically, we first define a CUSUM statistic at time $k$ as
\begin{equation} S_k = \max_{1 \leq n \leq k}\sum_{i=n}^{k}\ell(r_i), 
\label{eq:Def_Sk}
\end{equation} 
representing the maximum cumulative sum of LLRs over all possible intervals ending at time $k$. Intuitively, $S_k$ quantifies the strongest evidence that a change has occurred at some prior instant.
A change is declared as soon as $S_k$ exceeds a predetermined threshold $h$. The detection (stopping) time is thus given by 
\begin{equation} \tau(h) = \inf\{k \geq 1: S_k \geq h\}. \end{equation}

However, directly calculating $S_k$ in its definition (\ref{eq:Def_Sk}) requires computational effort proportional to $k$, as it involves evaluating the cumulative sums for all possible change-points. To circumvent this computational complexity, CUSUM leverages a simple recursive form. Specifically, the statistic $S_k$ can be equivalently updated using the following iterative procedure
\begin{equation}
g_k = \max(0, g_{k-1} + \ell(r_k)), \quad g_0 = 0.
\end{equation}
In practice, this recursive implementation dramatically simplifies computations and storage, enabling real-time, low-complexity detection. The detection rule then simplifies to
\begin{equation}
\tau(h) = \inf\{k \geq 1: g_k \geq h\}.
\end{equation}

The CUSUM algorithm possesses several desirable properties for online interference detection. 
First, it is optimal in terms of minimizing the worst-case expected detection delay under a given false alarm constraint. 
Second, its recursive form requires the minimal computational resources and memory, enabling efficient practical deployment. 
Last, due to its inherent reset mechanism (statistic reset to zero after crossing the threshold), the algorithm naturally offers robustness against transient observations and false detections, ensuring reliable and continual operation.

\subsection{Performance Analysis}
The performance of the CUSUM algorithm is typically evaluated using two metrics, i.e., the conditional average detection delay (CADD) and the false alarm rate (FAR).
The CADD measures the expected delay in detection under hypothesis $H_1$, defined as
\begin{equation}
    T_C(h) = \mathbb{E}_{H_1}[\tau(h)],
\end{equation}
where $\tau(h)$ is the stopping time when the detection statistic reaches the threshold $h$.

The FAR is usually expressed in terms of the reciprocal of the average run length (ARL) under the null hypothesis $H_0$, given by
\begin{equation}
    R_F(h) = \frac{1}{\mathbb{E}_{H_0}[\tau(h)]}.
\end{equation}

It is worth noting that the metrics of FAR and CADD are extensively employed in the literature to quantify the false alarm performance and the detection delay efficiency of change detection algorithms. In the following, we characterize the trade-off between FAR and CADD for the detectors proposed in this work.

\begin{theorem}
Consider observations $r_k$ obtained by projecting the received signals onto a known steering vector. Under hypothesis $H_1$, the observation $r_k$ follows a Rice distribution; whereas under hypothesis $H_0$, it follows a Rayleigh distribution. The asymptotic relationship between the CADD $T_C(h)$ and the FAR $R_F(h)$ as the detection threshold $h \rightarrow \infty$ is given by
\begin{equation}
    \lim_{h \rightarrow \infty}\frac{T_C(h)}{-\log R_F(h)} 
    = \frac{1}{I(\sigma)} 
    \approx \frac{\sigma^2 + a_N}{\sigma^4},
\end{equation}
where $\sigma = \sigma_I/\sigma_n$ represents the normalized interference-to-noise ratio (INR), and the constant $I$ denotes the average per-sample log-likelihood ratio under $H_1$
\begin{equation}
\begin{aligned}
    I(\sigma) &= \mathbb{E}_{H_1}[\ell(r)]\\
    &= \int_0^\infty \left[\log I_0\left(\frac{2\sigma^{(I)}r}{\sigma_n^2}\right) - \frac{\sigma^{(I)^2}}{\sigma_n^2}\right] f_{R}^{(1)}(r)\,dr.
\end{aligned}
\end{equation}

Moreover, we obtain the following tight bounds for this asymptotic relationship:
\begin{equation}
    \frac{\sigma^2 + 1}{\sigma^4} \leq \lim_{h \rightarrow \infty} \frac{T_C(h)}{-\log R_F(h)}
    \leq \frac{\sigma^2 + 3}{\sigma^4}.
\end{equation}
\end{theorem}

\begin{proof}
Proof is provided in .
\end{proof}

\noindent \textbf{Remark:} This theorem highlights the fundamental trade-off between detection delay and false alarm rate. Specifically, increasing the detection threshold $h$ results in an exponential reduction of the FAR but at the cost of a corresponding linear increase in detection delay. The factor $1/I$ quantifies the efficiency of this trade-off and is tightly bounded between $(\sigma^2+1)/\sigma^4$ and $(\sigma^2+3)/\sigma^4$. 
The proof leverages asymptotic properties of the CUSUM algorithm along with bounding techniques involving the log-Bessel function. 

\section{On-line Interference Change Detection with Unknown Direction}

In practical satellite communication systems, the direction of interference, denoted by $\theta^{(I)}$, is typically unknown, significantly complicating the detection task. This section addresses this realistic scenario by first developing a GLR-based detection algorithm and then introducing a suitable estimation method for determining the unknown interference direction.

\subsection{GLR-based Change Detector}

When the interference direction is unknown, the conventional CUSUM procedure described previously cannot be applied directly, since it requires prior knowledge of the steering vector associated with the interference source. To overcome this limitation, we adopt the GLR-based change detector, which replaces the unknown parameter(s) with their estimates.

Given a sequence of observations $\{\mathbf{y}_k\}$, the GLR-based detector examines all potential change-points to jointly detect interference onset and estimate the unknown interference direction. Specifically, at each time instant $k$, the GLR procedure considers all possible interference onset times $j<k$, computes the corresponding interference direction estimate, and evaluates the accumulated log-likelihood ratio from samples $\{\mathbf{y}_j,\dots,\mathbf{y}_k\}$. The GLR statistic at instant $k$ is thus defined as:
\begin{equation}
    G_k = \max_{1 \le j < k}\sum_{i=j}^{k}\ell\left(r_i(\hat{\theta}_{j,k})\right),
\end{equation}
where $\hat{\theta}_{j,k}$ denotes the estimated interference direction based on observations from time $j$ to time $k$, and the projection scalar measurement is 
\begin{equation}
    r_i(\hat{\theta}_{j,k}) = |\mathbf{a}(\hat{\theta}_{j,k})^H\mathbf{y}_i|.
\end{equation}

A detection event is declared at the first time instant when the GLR statistic $G_k$ exceeds a prescribed threshold $h$
\begin{equation}
    \tau(h) = \inf\{k \geq 1 : G_k \geq h \}.
\end{equation}

Having clearly established the GLR-based detection framework, the remaining challenge is the accurate estimation of the unknown interference direction. We now introduce and discuss the selected estimation method in detail.

\subsection{Direction Estimation Using Root-MUSIC}

Direction-of-arrival (DoA) estimation is a classical and extensively studied problem in array signal processing. Among various established algorithms are the maximum likelihood (ML) \cite{stoica1990}, MUSIC \cite{schmidt1986}, and Estimation of Signal Parameters via Rotational Invariance Techniques (ESPRIT) \cite{roy1989}. Given the ULA configuration adopted here and the requirement for computational efficiency, we employ the Root-MUSIC algorithm \cite{rao1989}, a closed-form variant of MUSIC specifically tailored for ULAs.

The Root-MUSIC algorithm estimates the interference direction through the following steps.

\begin{enumerate}
    \item Compute the sample covariance matrix using the $N$ most recent snapshots
    \begin{equation}
        \mathbf{R} = \frac{1}{N}\sum_{i=k-N+1}^{k}\mathbf{y}_i\mathbf{y}_i^H.
    \end{equation}

    \item Perform eigenvalue decomposition of the covariance matrix $\mathbf{R}$ to identify the noise subspace $\mathbf{E}_n$, consisting of eigenvectors corresponding to the smallest eigenvalues.

    \item Construct the Root-MUSIC polynomial based on the identified noise subspace
    \begin{equation}
        P_{\text{MUSIC}}(z) = \sum_{l = -(M-1)}^{M-1} c_l z^{-l},
        \label{eq:P}
    \end{equation}
    where
    \begin{equation}
        c_l = \sum_{i = 1}^{M - |l|}\sum_{j = 1}^{M - 1}[\mathbf{E}_n]_{i,j}[\mathbf{E}_n]_{i + |l|,j}^*.
    \end{equation}

    \item Obtain the roots of $P_{\text{MUSIC}}(z)$ in (\ref{eq:P}) and select the root nearest to the unit circle. The estimated interference direction at time instant $k$ is then obtained from the argument of this root
    \begin{equation}
        \hat{\theta}_k = \arcsin\left(\frac{\lambda}{2\pi d}\arg(z)\right).
    \end{equation}
\end{enumerate}

It should be noted that the proposed Root-MUSIC-based DoA estimator offers a closed-form solution with high-resolution accuracy and robustness to noise. 
This feature makes Root-MUSIC-based DoA estimator particularly suitable for real-time implementation in satellite communication systems, where computational resources are often constrained.

\section{Numerical Results}
\begin{figure}[t]
  \centering
    \includegraphics[width=0.9\linewidth]{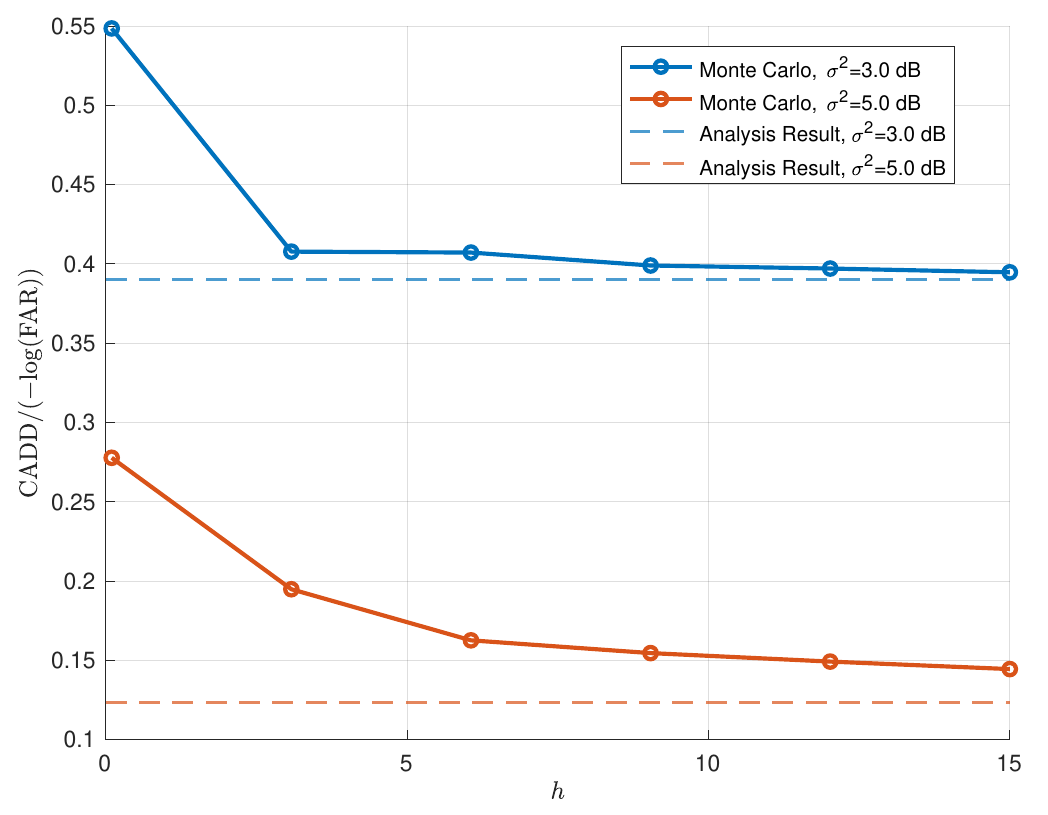}
  \caption{Numerical validation of the asymptotic performance of CADD and $-\log(\text{FAR})$ in Theorem 1.}
  \label{fig:Sim1}
\end{figure}

In this section, we present numerical simulations to evaluate and illustrate the performance of our proposed CUSUM-based and GLR-based interference detectors. Specifically, we consider a ULA consisting of $M=4$ antenna elements with an inter-element spacing fixed at half the wavelength. The system operates at a frequency of $1.6\,\text{GHz}$, which corresponds to the commonly utilized L-band frequency range in direct satellite-to-mobile communications (e.g., Starlink satellite network). Consequently, the potential interference signals simulated here also fall within this frequency band. In our simulations, the incoming signal direction is fixed at $90^\circ$ relative to the array axis (i.e., perpendicular incidence). 
Since the interference-to-noise ratio directly influences the performance of interference detection, we denote the interference-to-noise ratio by $\sigma^2$ in dB within our simulation scenarios. 
Additionally, to ensure statistically meaningful evaluations of the proposed algorithms, all the performance results presented here are obtained through extensive Monte-Carlo (MC) simulations, each involving over $10^{7}$ independent trials.

We first investigate the performance of the proposed CUSUM algorithm under the assumption that the interference direction is perfectly known. To illustrate clearly how interference energy impacts detection performance, Fig.~\ref{fig:Sim1} shows the ratio between CADD $T_C(h)$ and the negative logarithm of the FAR $-\log(R_F(h))$, evaluated at two representative interference-to-noise ratios: $\sigma^2 = 3\,\text{dB}$ and $5\,\text{dB}$. 
As we observe from Fig.~\ref{fig:Sim1}, the numerical simulations align closely with our theoretical predictions provided in Theorem~1. Specifically, with the increase in the detection threshold $h$, the simulated ratio $T_C(h)/[-\log(R_F(h))]$ consistently converges towards the theoretically derived value, validating the accuracy of the asymptotic analysis.
Furthermore, we clearly observe the expected trend that, under the same false alarm constraints, higher interference energy $\sigma^2$ leads to more effective and rapid detection. 
This indicates that stronger interference signals are naturally easier to detect and thus yield shorter detection delays. These observations corroborate the theoretical intuition highlighted previously and confirm the practical effectiveness and robustness of our proposed CUSUM-based interference detection algorithm in regimes of practical interest.

\begin{figure}[t]
  \centering
    \includegraphics[width=0.9\linewidth]{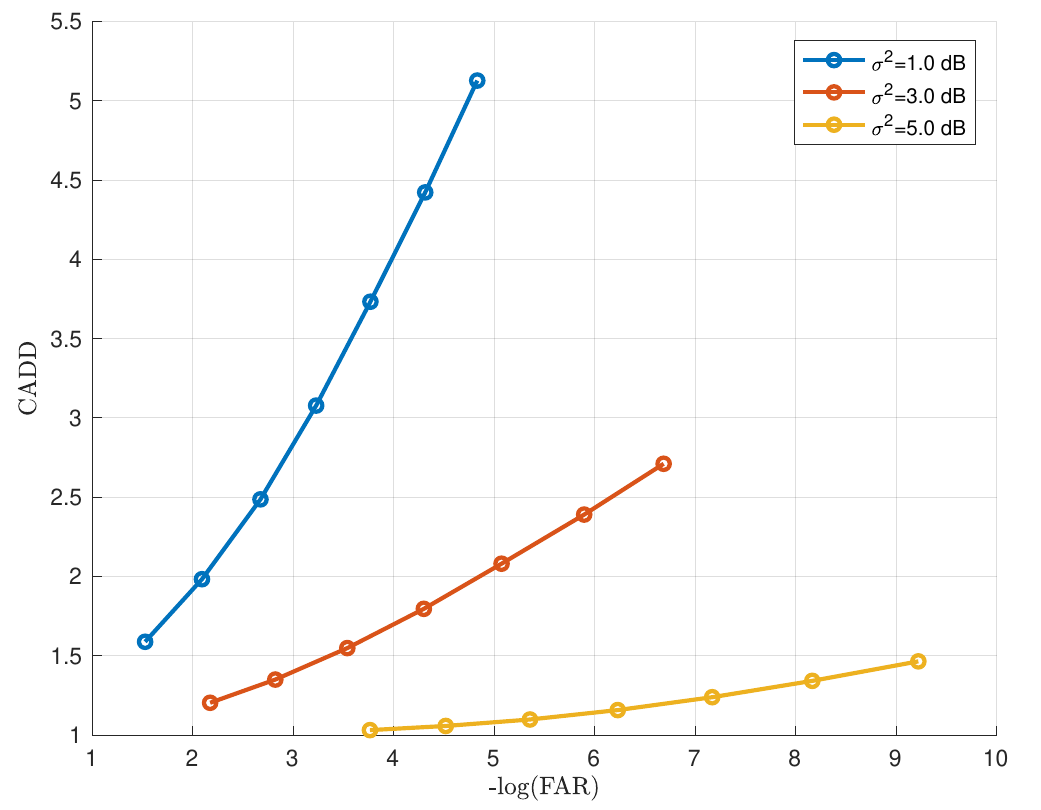}
  \caption{Performance of the CUSUM-based detector: Effective detection under low interference-to-noise ratios ($\sigma^2=1,3,5\,\text{dB}$).}
  \label{fig:Sim2}
\end{figure}
Next, Fig.~\ref{fig:Sim2} presents the trade-off between the CADD and negative logarithm of the FAR ($-\log(\text{FAR})$) of our proposed CUSUM-based interference detection algorithm for various interference-to-noise ratios. 
To further highlight the algorithm’s capability, we specifically consider scenarios where the interference is relatively weak.  
Indeed, as shown in Fig.~\ref{fig:Sim2}, even when the interference energy is merely $1\,\text{dB}$ higher than the noise level (i.e., a subtle and challenging detection scenario), the proposed CUSUM detector effectively controls the expected detection delay to below 3 samples for $-\log(\text{FAR})=3$. 

Next, we focus our attention on evaluating the performance of the generalized likelihood ratio (GLR)-based detection algorithm, designed explicitly for scenarios in which the interference direction is unknown. 
Fig.~\ref{fig:Sim3} shows that, even without prior knowledge of the interference direction, the GLR method is capable of effectively detecting interference, though inevitably with somewhat higher detection delays compared to the known-direction CUSUM algorithm. 
This increase in delay arises naturally from the additional uncertainty involved in jointly estimating the interference direction and performing detection. Specifically, under the scenario of weak interference presence ($\sigma^2=1\,\text{dB}$), the GLR detector can maintain an average detection delay around 6 samples at $-\log(\text{FAR})=3$, highlighting the algorithm’s robustness and practical utility in challenging real-world conditions.
Moreover, simulation results in Fig.~\ref{fig:Sim3} also confirm an intuitive and expected trend: as interference energy increases, the GLR detector achieves significantly lower detection delays, thereby demonstrating improved performance. For instance, when the interference-to-noise ratio reaches $4\,\text{dB}$, the detection delay becomes consistently small (approximately 2 samples), even under strict false alarm constraints. 
This performance enhancement illustrates that, despite the inherent complexity of unknown-direction scenarios, strong interference signals can still be detected promptly and reliably through the GLR-based approach.

\begin{figure}[t]
  \centering
    \includegraphics[width=0.9\linewidth]{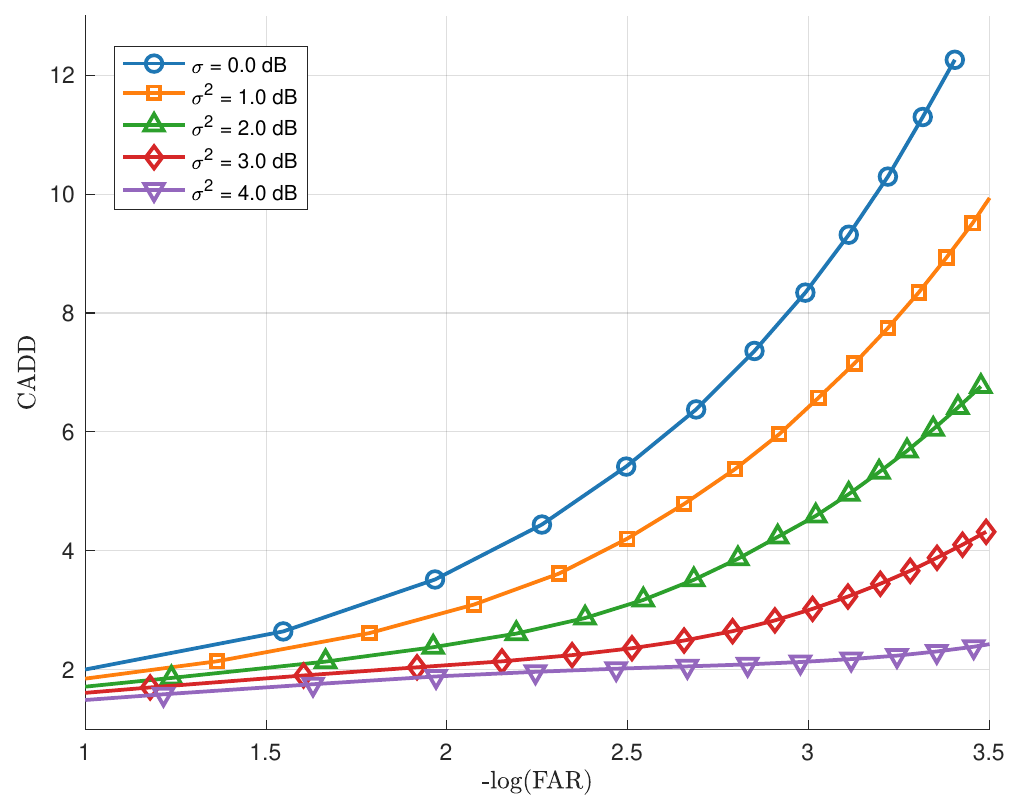}
  \caption{CADD versus $-\log(\text{FAR})$ for the GLR-based interference detector under various interference levels ($\sigma^2=0,1,2,3,4\,\text{dB}$).}
  \label{fig:Sim3}
\end{figure}

\section{Conclusion}
In this paper, we investigated online interference detection methods tailored for direct satellite-to-device communications. Exploiting the periodic idle intervals inherent in burst-mode satellite transmissions, we formulated the interference detection as a binary hypothesis testing problem based on Rayleigh and Rice signal distributions. We proposed a CUSUM-based detector for cases with known interference directions, establishing analytical optimality results. For practical scenarios with unknown interference directions, we further developed a GLR-based detector together with Root-MUSIC direction estimation. Numerical results confirmed the effectiveness of our proposed methods, demonstrating rapid detection with high accuracy, and thus validating their feasibility for practical satellite-to-device communication systems.

\bibliographystyle{IEEEtran}
\bibliography{main} 

\end{document}